\documentclass[a4paper, 12pt]{article}
\usepackage{amsthm}
\usepackage{amsmath,amssymb}
\usepackage{graphicx,csquotes,xcolor,epigraph}
\usepackage{authblk}
\usepackage{xcolor}

\usepackage[margin=0.5in]{geometry}

\newtheorem{theorem}{Theorem}[section]
\theoremstyle{plain}

\theoremstyle{definition}

\usepackage{amsthm}
\newtheorem*{thm}{Theorem}

\begin{document}

\author[1]{\small Gustaf Arrhenius}
\author[1, 2]{\small Klas Markstr\"om}

\affil[1]{\footnotesize Institute for Futures Studies, 101 31, Stockholm, Sweden}
\affil[2]{\footnotesize Department of Mathematics and Mathematical Statistics, Ume\aa\ University ,SE-901 87 Ume\aa , Sweden}

\title{More, better or different?\\
Trade-offs between group size and competence development in jury theorems}
\maketitle

\begin{abstract}
	In many circumstances there is a trade off between the number of voters and the time they can be given before having to make a decision since both aspects are costly. An example is the hiring of 
	a committee with a fixed salary budget: more people but a shorter time for each to develop their competence about the issue at hand or less people with a longer time for competence development?  In this paper we investigate the interaction between the number of voters, the development of their competence over 
	time and the final probability for an optimal majority decision.  Among other things we consider how different learning profiles, or rates of relevant competence increase,  for the 
	members of a committee affects the optimal committee size. 
	
	To the best of our knowledge, our model is the first that includes the potentially positive effects of having a heterogeneous group of voters on majority decisions in a satisfactory way. We also discuss how 
	some earlier attempts fail to capture the effect of heterogeneity correctly.

\end{abstract}


\section{Introduction}
The various perceived benefits of collective decision making were described very early on. Aristotle, for example, writes:

\blockquote{
For the many, who are not as individuals excellent men, nevertheless can, when they have come together, be better than the few best people, not individually but collectively, 
just as feasts to which many contribute are better than feasts provided at one person's expense. For being many, each of them can have
some part of virtue and practical wisdom, and when they come together,
the multitude is just like a single human being, with many feet, hands, 
and senses, and so too for their character traits and wisdom. That is why
the many are better judges of works of music and of the poets. For one of
them judges one part, another another, and all of them the whole thing.}{--Aristotle,  Politics \cite{Reeve}}

First out in this description is the idea that a collective decision could be taken with greater competence than that of any of the individual decision makers. Indeed, this is the positive half of Condorcet's celebrated jury theorem:
\begin{theorem}[Condorcet's jury theorem]
	Given an odd number, $n$, of independent jurors, each of which votes for the correct verdict with a fixed probability $1>p>1/2$,   the probability $P_n$ of a correct verdict grows monotonically to 1 as $n\rightarrow \infty$
\end{theorem}

However, there are some aspects already in Aristotle's' description of the benefits of collective decision making which do not fit into Condorcet's theorem in its original form, and also situations where an application of the theorem yield conclusions which contradict practical experience. An example of the first aspect is given by Aristotle's emphasis on how a variety among the decision makers can improve their collective competence. In the basic form of the jury theorem there is no room for such effects since all individual competences are fixed and the same. Moreover, as we shall soon see,  the existing attempts to incorporate this aspect in generalised jury theorems still fall short. An example of the second type comes from the composition of committees. The application of the jury theorem to the question of how to set up a committee in order to get the highest probability of the correct decisions yields that we should make the committee as large as possible as our financial resources allow. That committees simply become better the larger they are doesn't square well with empirical observations and, as for example \cite{Fr82} has shown, for large enough committees, the internal structure becomes crucial for how competent it will be.

We claim that the missing component in the various versions of the jury theorem is the fact that typically some amount of time  will pass from the  point when an election is announced, or a committee is formed, to the point when votes are cast, and this gives room for competence change among the voters. How the individual competence improves during this time, either by an individual's own work, deliberation, or by different forms of interaction with the other voters in the group, is affected by the variety of voter background, as in Aristotle's description above, and influences the optimal structure of an efficient committee.  

We will first take a quick look at how Condorcet's original theorem has been generalised in order to cover more realistic situations, and then outline a way to include developing competences over time in the theorem. 

The conditions in Condorcet's theorem are quite strict: the jurors are independent and $p$ is fixed and equal for all jurors.  Dietrich and Spiekerman \cite{DS17,DS23} provide good surveys and additional critiques of these assumptions. However, one can easily show that these conditions can be relaxed substantially.  We can have jurors numbered $i=1,\ldots, n$  and let each have an individual probability $p_i$.  If we now let $p$ denote the average $p=\frac{1}{n}\sum_{i=1}^n p_i$  then the theorem still  holds \cite{Bo89}.   We also do not need to keep $p$ independent of $n$, a number of different theorems on how close to its mean a random variable is likely to be, imply that as long as $(p-1/2)\sqrt{n}\rightarrow \infty$ the theorem still holds \cite{BP98}.  That is, as the number of voters grow we can allow $p$ to be just a bit larger than $1/2$  plus 1 over the square root of the number voters.

Finally, and perhaps least well-known,  it is not hard to show that instead of requiring that the jurors are independent, it is enough to require that the average size of their pairwise correlations is not too high. This was shown already in \cite{La92}, but has also been considered in more recent papers \cite{KZ11}.  One important corollary here is that  \emph{negative} correlations are in fact beneficial for the probability $P_n$, rather than a problem, and the more negative they are the better. 

Some papers \cite{SK17} have tried to use variation in individual probabilities together with correlations to explain the benefits of a heterogeneous electorate in some specific circumstances.  The focus in those papers was on the elimination of the detrimental effects of underlying biases in the electorate, rather than on the positive effects of different forms of background competence.

\section{Earlier results}
In this section we will review some relevant existing variations of Condorcet's original jury theorem, including some which aims at balancing the cost for salaries versus the cost incurred by incorrect decisions. However, unlike the results we shall show, the individual competence in all these results is static.

We will here use $X_i$ both to refer to the $i$th voter, and the 0/1-valued  random variable which is the vote of that voter.  For  $X_1, X_2, \dots$ we assume  that $\mathrm{Pr}(X_i=1)=p_i$
 and  $\mathrm{Pr}(X_i=0)=1-p_i$.  We let  $Z_n=\sum_{i=1}^n X_i $.  We assume that a value of 1 is a vote for the correct alternative, so for  a simple majority decision we are interested in  $\mathrm{Pr}(Z_n>\frac{n}{2})$. We will also use the following notation:  ${\bf p}=(p_1,p_2,\ldots,p_n)$ and ${\bf p}(p)=(p,p,\ldots,p)$.  Note that, unless we specify otherwise, the vectors ${\bf p}$ for different $n$ are not correlated in any way. 

For independent voters some of the basic variations of the original theorem are included in the following theorem, for which we include a proof in the appendix. We will throughout the paper assume that $n$ is odd, but the results hold for even $n$ as well if a tie-breaking rule, such as a fair coin-toss, is added. 

\begin{theorem}
	Let $P_n(\bf{p})$ denote the probability of a majority for the correct outcome and define the average competence as $\bar{p}=\frac{1}{n}\sum_{i=1}^n p_i$.
	\begin{enumerate}
		\item  If $p>1/2$ is  fixed  then $P_n({\bf p}(p)) \rightarrow 1$ monotonically with $n$. 
		
		\item   If $ \bar{p}= p> \frac{1}{2} $ for some fixed $p$ then $P_n \rightarrow 1$ and $P_n({\bf p })\geq P_n({\bf p}(p))$.

		\item  If $\bar{p}=\frac{1}{2}+\frac{\omega(n)}{\sqrt{n}}$, where $\omega(n)$ is any increasing, unbounded, function of $n$, then $P_n \rightarrow 1$.
		
	\end{enumerate}	
\end{theorem}
Part 1 here is the classical Condorcet jury theorem. As has often been pointed out this basic version is based on stronger  assumptions than any real life situation is likely to satisfy, see for example \cite{dietrich_2008} for an in depth discussion of this. Part 3 is a version of the result from \cite{BP98}.

Part 2  is a combined version of results from \cite{Bo89,Owen89}, which established that one can use the value of $\bar{p}$ instead of a homogeneous value $p$, and \cite{KANA}, who proved that the homogeneous situation with probabilities given by $\bf{p}$ actually gives the lowest probability for correct decision. At a first glance this might be interpreted as a confirmation of the idea that heterogeneity is desirable, but this is a misleading interpretation.  What the statement says is that if we take two juries with exactly the same size and value for $\bar{p}$, where one is the homogeneous jury given by ${\bf p}(\bar{p})$ and the other by some heterogeneous ${\bf p}$, then the latter will have a higher probability for a correct decision. However the mechanism behind this is really based on the fact that the heterogeneous jury must, in order to be both heterogeneous and have the same $\bar{p}$, have several members with higher competence $p_i$ than $\bar{p}$, and their influence on the probability outweighs the effect the low competence members. 

In fact, an earlier theorem by Hoeffding \cite{hoeffding1956} identifies the exact jury composition which maximises the probability for a correct decision, with a given $\bar{p}$. This is given by having $\lfloor \bar{p}  n \rfloor$ members with $p_i=1$, one with $p_i=  \bar{p}  n -\lfloor \bar{p}  n \rfloor$, and the remaining $n-\lfloor \bar{p}  n \rfloor$ with $p_i=0$.  So for $n=3$  and $\bar{p}=0.7$ we could  have had a jury with ${\bf p}=(0.7,0.7,0.7)$, giving a non-zero probability for an incorrect decision,  and Hoeffding's theorem shows that for this value of $\bar{p} $ the maximum probability for a correct decision is reached by a jury with ${\bf p}=(1,1,0.1)$.   For $\bar{p}>1/2$ a  jury of the form identified by Hoeffding  \emph{always} makes a correct  majority decision, but one can hardly claim that this is due to a positive de-homogenising effect by those jurors which have $p_i<1$. Hoeffding's theorem was further refined by Glessel \cite{Gle75} in a way which provides a simple condition for deciding which of two jury compositions lead to the highest probability for a correct decision, a result which is applicable for any finite size $n$.   With two juries given by ${\bf p }^a$ and ${\bf p }^b$  we say that jury $a$ \emph{majorizes} jury $b$  if $\sum_{i=1}^j{\bf p }_i^a \geq \sum_{i=1}^j{\bf p}_i ^b$ for every $j=1, 2,\ldots n $, and Glessel's result then states that if $a$ majorizes $b$ then $a$ has the higher probability for a correct decision.

One can also obtain valid forms of the jury theorem for situations where the jurors are no longer independent. The following theorem was proven by Ladha in \cite{La92}. 
Recall that by standard definitions $\mathrm{Var}(X_i)=p_i(1-p_i)$ and $\mathrm{Cov}(X_i,X_j)=\mathbb{E}(X_i X_j)-p_ip_j$.
\begin{theorem}\label{corrthm}
	We use the same terminology as in the previous theorem but now we allow the $X_i$ to be dependent. 
	Let $\bar{p}=\frac{1}{n}\sum_{i=1}^n p_i$,  $d=n(\bar{p}-1/2)$,  and 
	$$\sigma^2=\sum_i \mathrm{Var}(X_i)+2\sum_{i<j}\mathrm{Cov}(X_i,X_j)$$.
		
	Then $P_n({\bf p}(p)) \geq   \frac{d^2}{\sigma^2+d^2}$
\end{theorem}
Note that this theorem only requires knowledge of the pairwise correlations among the $X_i$, even though knowledge of all $k$-wise correlations are needed in order to fully reconstruct a correlated distribution.  The price of only looking at the pairwise correlations  is of course that the bound in the theorem is sometimes far from optimal, in the sense that  the probability ${\bf P_n}$ can be higher than what  the theorem guarantees. Adding  additional assumptions  about the the joint distribution for the $X_i$  can easily improve the bound.  An extreme example is given by the distribution which sets exactly $\lceil \frac{n}{2}\rceil$ of the variables to 1, with all such assignments given equal probability. With this distribution there is always a correct majority decision, while the bound in the theorem goes to 0 as $n$ grows. 

As a general rule we see that positive correlations reduce the bound for a correct decision and negative correlations improve the bound, the latter is the explanation for the example we just gave. At the same time, since we are ignoring higher order correlations,  one can construct examples of distributions with different correlations and the same probability for a correct decision \cite{Kani10}
An interesting question here is whether we can create negative correlations, or reduce the positive ones,  in a jury by e.g. selecting jury members which connect to  different basic moral foundations \cite{GHN09}. We will return to this, and correlations in general,  in Section \ref{corrsec}.

There has also been generalisations of the jury theorem which add additional elements to the set up, apart from the individual competencies.   In \cite{SK17} a version where each juror has some inherent biases is considered  and it is demonstrated that when biases are strong  the composition of the committee can strongly influence probability for a correct decision.   In \cite{LG10} costs are added to the problem of selecting a committee. This is done by selecting members from a pool of candidates each of which have both a known individual competence and a salary cost,  and also assigning  a cost to incorrect decisions.   Here it is shown that the expected total cost, for both salaries and  incorrect decisions, is minimised  for some committee composition which typically consists of a much smaller committee than the full pool of potential members.

\section{Competence and time}

Our basic set up is as follows: We have $n$ voters $i=1,\ldots,n$ each with an individual competence $p_i(t)$ regarding the issue at hand.  Here the individual competence $p_i(t)$ depends on the time $t$, with $t=0$ corresponding to the point where the voter is made aware of the issue, and some later time $t=T$ being the time at which the vote is held.   In this setting the group competence $P(t)$ is also time dependent and the way it develops will depend both on how the individual competencies $p_i(t)$ can be improved over time and on how correlations among the voters develop.

In this first discussion we will make the simplifying assumption that correlations are negligible in the final vote, and instead focus on the effect on different learning profiles for the individual competencies and the interaction between the number of voters, as well as the total cost of the whole process. By a learning profile we here mean the average individual competence $p(t)$ as a function of time. We will first discuss the behaviour of majority decisions for a specific simple form of learning profile, and then discuss other types of learning profiles and under which circumstances they are likely to occur.

\subsection{Individual  vs group competence}
The exact connection between a group's effective competence   and that of its individual members is a debated issue, see e.g. \cite{Kallestrup,pino2021group,BBB}. For the mathematical form of our result the details of this is not important, only the individual's effective competence matters. However, we will first look at two different  ways in which this effective competence may come about.

In our results we let $p_i(t)$ denote the competence of individual number $i$  on the issue to vote on, at time $t$.     When the individuals are part of a group performing 
some kind of deliberation on the issue at hand  this can have two effects on  $p_i(t)$. First, the competence of $X_i$ can increase when $X_i$ is regarded in isolation.  This is of course the effect we see in a one-person committee, where the entire improvement  in  $p_1(t)$  comes from the improvement  in the competence of that individual.

Second,  we may also see an improvement in  $p_i(t)$  which is only present as long as individual $X_i$ is part of the group.  As an example, let us assume that we gathered a committee to make a decision on the construction of a railway bridge over a gorge. Here we may have a geologist, an expert on explosives, a railway engineer, someone in charge of the budget, and several additional experts.   During the deliberation on how to design and build this bridge each committee member is likely to learn new things which raise $p_i(t)$ by some amount. In addition to this lasting increase of their individual competences, they may also gain an effective increase in their  $p_i(t)$  which is group dependent rather than lasting.  For example, the railway engineer will be able to dismiss  some infeasible designs thanks to the knowledge of the geologist, and will be able to make additional improvements as long as the geologist is part of the committee. This added effective competence will however mostly be lost when the committee is dissolved, since the engineer is unlikely to have learnt all the relevant knowledge of the geologist.  

A real-life example of many of these issues can be found in the Polymath projects, an online collaboration aimed at solving some open problems in mathematics. The original Polymath project started in the blog  of the Fields medalist Timothy Gowers \cite{G09} and successfully solved the problem suggested by him, leading to two published papers. The collaborative aspects of the projects have been analysed in several papers \cite{B10,CK11}.


\subsection{Improving competence over time}
Here we shall also assume that from the point when a committee is formed or an election is announced, the voters will undertake deliberation or other learning activities which improve their competence of the issue at hand.  This learning activity can take on many different forms. For a single individual this can be some form of individual study and investigation, as must be the case for a single person committee. Those activities are also available for members of a larger group, but here some form of group deliberation may also take place. 

We will refer to the way $p_i(t)$ changes as a \emph{ learning profile}. This  term only refers to the change in the value of $p_i(t)$, not the process behind how that change comes about. 

As we will see there are  some resource-constrained situations in which a large group can outperform a smaller one  if the learning profile for the larger group is faster than that for the smaller group. This can happen trivially if the members of the larger group are simply better at learning than those in the smaller group, but a more interesting situation is when this instead comes about as a genuine group  advantage. For example, by bringing in different  background competencies, as in the bridge example above, or by delegating different parts of the fact-finding process to different members, the group might improve the joint competence in a more efficient way.

In our results we will focus on the learning profile and see how different learning profiles affect the group competence.  A question which we leave open is how different modes of collaboration and deliberation give rise to different learning profiles. 

\section{The probability for a correct decision by a simple majority vote}
Here we quickly recall a few basic fact about the probability $P(n,p)$  that a group of $n$  independent individuals with competence $p$  reach a correct  decision, when voting  under 
unweighted majority.
$$P(n,p)=\sum_{\lceil n/2\rceil}^{n}  p^i(1-p)^{n-i}{n \choose i} $$
The function $P(n,p)$ has some useful properties:
\begin{enumerate}
	\item  For $p\geq 1/2$,  $P(n,p)$ is a concave increasing function in $p$.	
		
	\item  For $p\geq 1/2$,  $P(n,p)$ is an increasing function in $n$.
	
	\item For large $n$,  we have that $\frac{d}{dp}P(n,1/2) \sim \sqrt{n2/\pi}$.	
	
\end{enumerate}

In order to help guide intuition, in In Figure \ref{probs}   we display $P(n,p)$ for a few values of $n$ and $p$ in the interval $[\frac{1}{2},1]$

\begin{figure}
  \begin{center}
    \includegraphics[width=0.48\textwidth]{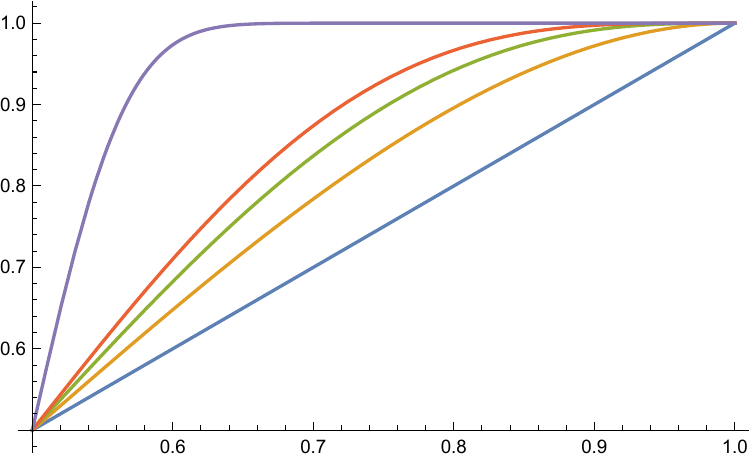}
  \end{center}
  \caption{The probability for a correct decision as a function of $p$,  for $n=1,3,5,7,91$.}\label{probs}
\end{figure}

\section{Linear learning profiles}
Here we will first examine one of the simplest non-constant learning profiles.  We let $p(t)=1/2 +c t$, where $c$ is a constant, for $t<1/c$, and $p(t)=1$ for $t\geq 1/c$. Here we see a linear improvement in the average competence until the average reaches 1, and then the competence stays at 1. In most situations this is an unrealistically efficient learning profile, but it gives rise to the same qualitative behaviour as more realistic ones, and makes the mathematical analysis easier to follow.   In some cases we will use different values of $c$ for different $n$ and then denote this value by $c_n$.

\subsection{Fixed total time}
Let us now consider the situation where we have a fixed total amount of time $T$, and we can either let one voter use the whole amount or instead let $n$ voters use $T/n$ each.   This would  e.g. correspond to the situation where we are setting up a committee. We have a given budget for salaries and are free to spend that budget on either a one member committee, working for a longer period of time, or a larger committee which has to finish earlier.   To make things explicit we will first take $n=3$.

We first take $c=1$. In the leftmost part of Figure \ref{fig1v3a} we display the group competence for the two group sizes a a function of $T$.  Here we can clearly see that unless $T$ is much larger than here, a single voter will achieve a higher group competence than 3 voters. If $T$ is so large that $p(T/3)=1$ then the group  competence for three voters will have caught up with that for a single voter.

\begin{figure}
  \begin{center}
  	\includegraphics[width=0.48\textwidth]{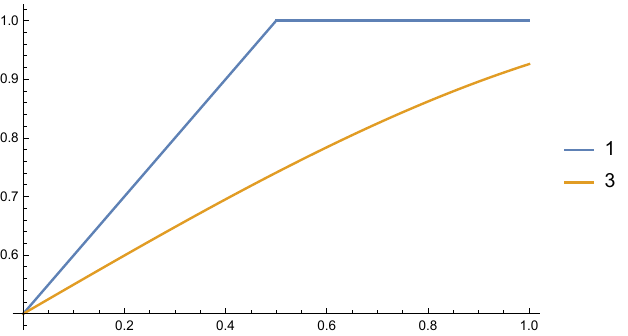}
	\includegraphics[width=0.48\textwidth]{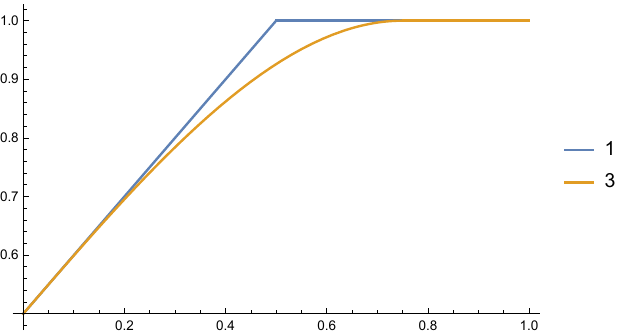}
  \end{center}
  
  \caption{Group competence  for $c_1=1$   and $c_3=1$ and $2$  }\label{fig1v3a}
\end{figure}	
\begin{figure}
  \begin{center}
	\includegraphics[width=0.48\textwidth]{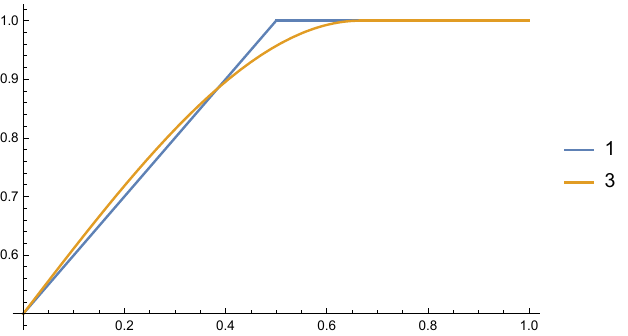}
	\includegraphics[width=0.48\textwidth]{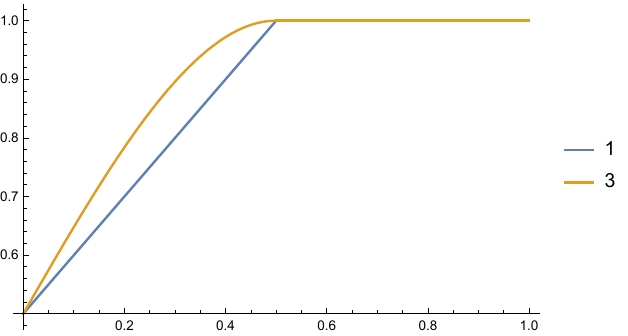}
  \end{center}
  
  \caption{Group competence  for $c_1=1$   and $c_3=2.25$ and $3$  }\label{fig1v3b}
\end{figure}

That this is the case can easily seen by calculating the derivative of the group competency with respect to $T$.   For $n=1$  the derivative is simply $c$.    For $n=3$ we take the derivative of 
$P_3(T)=(3 - 2 p(T/3)) p(T/3)^2$ which after a bit of algebra is $c/2 - (2 c^3 T^2)/9$.    At the point $T=0$ we thus have derivative $c$ for $n=1$  and $c/2$ for $n=3$.  Since the latter is smaller, and P(T) is a concave function of $T$,  the group competence for $n=3$ will remain smaller than that for $n=1$, until  $p(T/3)=1$.
 
So, with a bounded total time and a linear learning profile a single voter has the advantage as long as those in  the 3-member committee learn at the same rate as an isolated individual. 
We can instead look at what happens if the larger group now has a learning profile of the form $p_3(t)=1/2+c_3 t$ for $t\leq 1/c_3$. As long as $c_3/2\leq 1$ the previous argument still applies, and the single voter has the advantage.   For $c_3=2$ the two functions are tangent at $t=0$, but the ordering remains the same, with the single voter performing better than the group.  For $2<c_3<3$  the larger group outperforms the single voter for an initial range of $T$, and at a large value of $T$ the single voter regains  the advantage.    For $c_3>3$  the larger group has the advantage for all values  of $T$.  In Figures \ref{fig1v3a} and \ref{fig1v3b} we display the two functions for several different values of $c_3$.

Here we see that with a fixed time budget the larger group performs worse than a single voter unless the larger group actually takes advantage of the group to improve the learning rate, and this improvement in learning rate must be sufficiently large in order outperform a single voter.   So, we do not see any   "wisdom of the crowd" merely by having a crowd, communication is necessary.

The corresponding  critical values  $c^*_n$, after which a group of $n$ voters can perform better than a single voter,   for $n$ from 3 to 15 are  $${2, 8/3, 16/5, 128/35, 256/63, 1024/231, 2048/429}$$  so a group with $n=7$ members has to learn more than three times as fast as a single voter, and a group with $n=11$  more than four times as fast.

For larger groups the demand on the learning rate $c_n$ can be found asymptotically. For large $n$ the derivative of $P_n(t)$ is to leading order given by $\sqrt{n2/\pi}$ and when setting $t=T/n$ this gives a total derivative at $T=0$ of  $c_n \sqrt{\frac{2}{n  \pi}}$.   Thus, in order for a  group of size $n$, when $n$ is large,  to perform better than a single voter we must have $c_n>\sqrt{\frac{n  \pi}{2}}$.     Hence, the larger the group is, the more they need to be able to take advantage of other group members in order to improve the  average competence of the group.

\subsection{The cost of reaching a given group competence}
Instead of looking at which group competence we can reach with a given budget for the total time, we can also consider the cost  $C(n,P^*)$ of achieving a given group competence $P^*$ for 
different group sizes.    The function $P_n(p)$ is approximately linear as a function of $p$ for $|p-1/2|\leq \alpha/\sqrt{n}$, where $\alpha$ is a constant smaller than 1, and using this we can estimate the cost of reaching a given group competence.

In order to find the total cost for achieving a given group competence $P^*$ with $n$ voters we first find the value $p^*$ which gives $P_n(p^*)=P^*$, second we find that time $t^*$ such that $p_n(t^*)=p^*$  and finally compute the cost as $n t^*$.    Here we see that again the learning profile enters the cost, and due to the linear nature of $P(t)$ for small $t$ it turns out that the cost $C(N,P^*)$ behaves differently depending on whether the learning rate is below, on, or above, the critical learning ratio $c^*_n$ which we have already identified.    For learning rates  $c_n$ smaller than $c^*_n$ the cost grows with $n$ and is unbounded;  at the critical learning rate the cost converges; and for learning rates larger than the critical one, the cost is decreasing with $n$. In Figure \ref{cost-1} we show examples of the three different behaviours.  

\begin{figure}	
  \begin{center}
  	\includegraphics[width=0.40\textwidth]{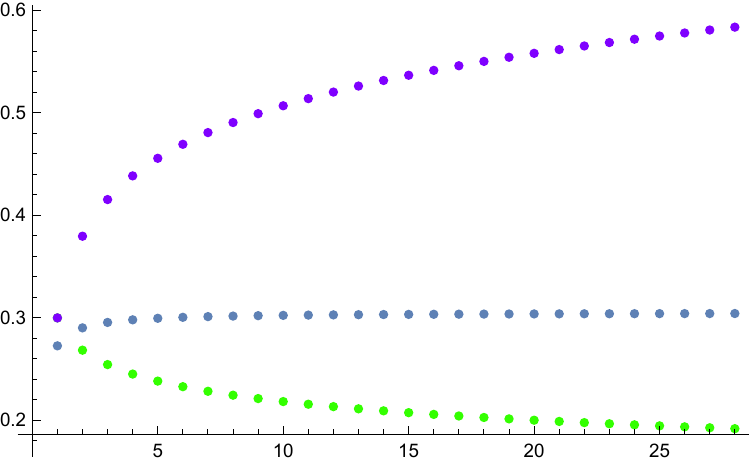}
  \end{center}
  
  \caption{The cost for reaching a given group competence, here $P^* =0.8$, as a function of $n$, for three different learning rates.}\label{cost-1}
\end{figure}

So, the critical learning rate $c^* _n$ characterises both when a larger group is able to perform better under a given time budget, and when a large group can achieve a given group competence at a lower cost than a single individual.

\subsection{Compensating for slow learning by increasing the group size}
In a setting where the time until a decision must be made is fixed  we may also ask how a single highly efficient expert investigator performs  compared to a larger group of non-experts.

Here we assume that the expert has learning profile $p(t)=\frac{1}{2}+c_1 t$ and the members of the larger group learn at a standardised rate $p_n(t)=\frac{1}{2}+ t$.

If $c_1=1$ then clearly the larger group will perform better, this is just the conclusion of the classic jury theorem, and any increase in $c_1$ will improve the performance of the expert. So a natural way to phrase this problem is to ask, from which value  $c_1(n)$ will the expert out-perform a group of $n$ non-experts, for small $t$?   The reason for first considering small $t$ is that unless we know more about the time until decision we cannot rule out situations like the one in the left part of Figure \ref{fig1v3b}, where the group initially has the higher competence but the single voter eventually overtakes the group.

For small values of $n$  this can be calculated directly by differentiating $P_n$ with respect to $t$.   For $n$ from 3 to 15 we then get the values of $c_1(n)$ as:
$$ 3/2, 15/8, 35/16, 315/128, 693/256, 3003/1024, 6435/2048$$
So, for $c_1(3)>3/2$  the single expert will outperform a group of three non-expert voters, while  for $n=7$ we find $c_1(7)=35/16$, and so the expert must learn more than $2.18...$  times as fast as the seven non-expert voters.

For larger  $n$ we can use the fact that we know the derivative of $P_n(t)$ to leading order  to get the approximation  $c_1(n)=\sqrt{\frac{n2}{\pi}}+o(1)$, where the $o(1)$-term is a positive error term converging to 0.

\subsection{The high competence range}
In each of the situations considered so far we focused on the range where the group competence  $P(n,p)$ is roughly linear in $p$.   This is relevant for situations where we begin with an average individual competence close to 1/2  and time does not continue long enough to reach competence much larger than $p=1/2+C/{\sqrt{n}}$, for some constant $C>0$.  If $t$ is allowed to be much larger than this  then the group competence will reach a region where it is close to 1 and almost constant, and so quite insensitive to further improvements in $p$. This corresponds to  the nearly horizontal part of the graphs in Figure \ref{probs}.

There is of course also a middle range where $p$  is close to where $P(n,p)$ changes from growing linearly in $p$ to being nearly constant. This  is the region close to $p=0.6$ for $n=91$ in Figure  \ref{probs}.   In this region the linear approximation is no longer valid and in order to see where a larger or smaller group has the advantage we have to make a calculation for those two specific values of $n$, rather than using a simplifying approximation.    Given the relatively simple form of $P(n,p)$ this can quickly be done using computer algebra, for any concrete pair of group sizes and explicit learning profiles. Hence, our discussion already covers the range where results can be given a simple explicit form.

\section{Learning profiles}
So far our discussion has focused on linear learning profiles and while these provide clear examples, and their properties are representative for many more general profiles as well,  they are not likely to be exact models for the growth  of competence in real life cases.   In this section we will discuss both how general learning profiles can differ from the linear case and which type of learning profile we might see in different settings.

\subsection{General learning profiles}
For completely general learning profile very little can be said, since this allows e.g. an oscillating  mean competence. However, we can identify some general features of certain classes of learning profiles.

For learning profiles which are continuous and weakly monotone increasing to 1 (i.e. allowed to remain constant for some time intervals but not to decrease) and do not depend on the number of voters $n$, we see the same type of behaviour as in the linear case.  For profiles of this type we can rescale the time  and map the group competence to that of a linear profile, with a more complicated time dependency for the larger group.  Here the qualitative behaviour, when comparing a single learning profile for different numbers of voters,  is the same as for a linear profile, but the time dependency can be more complicated.

If the form of the learning profile depends on the number of voters  we can get stronger versions of some of the behaviours which we have seen in the linear case.    Let us compare a  linear profile for a single voter with a concave profile for several voters. As an example we can take $p_1(t)=\min(1/2+t,1)$ as the individual competence for the single voter and $p_3(t)=\min(1/2+t^{0.55},1)$ as the individual competence in a group of three voters. Here the individual competence in the larger group is a concave function of time, leading to  rapid growth for small values of $t$ and then a much slower growth for larger $t$.  On the left  in Figure \ref{lin-con} we display both the individual competences and the group competence as a function of the  total time $T$.   As we can see, the group competence for the group of three voters grows rapidly  and remains higher than that for the single voter until the single voter has managed to reach a very high competence and overtakes the larger group, due to the decreasing learning rate in the concave profile.    Here the preference between a larger and smaller group depends strongly on the time budget.  How distinct the behaviours of the two profile are depends on how concave the learning profile for the larger group is. In our simple example replacing the exponent $0.55$ by number closer to 1 will move us towards the linear case and making the exponent smaller will quickly make the profile even more strongly concave.   If we  have a convex profile, e.g. taking the exponent in our example to be $2$, then the advantage for the single voter will instead increase, as shown in to the right of Figure \ref{lin-con}.

\begin{figure}	
  \begin{center}
  	\includegraphics[width=0.40\textwidth]{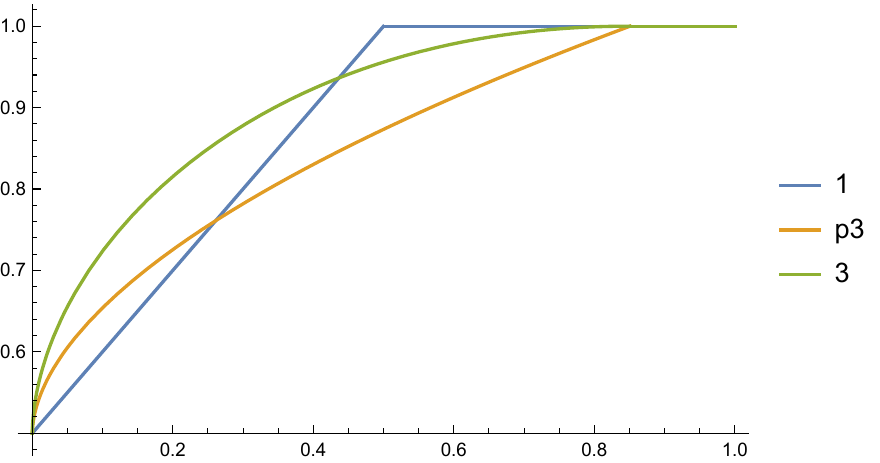}
	\includegraphics[width=0.40\textwidth]{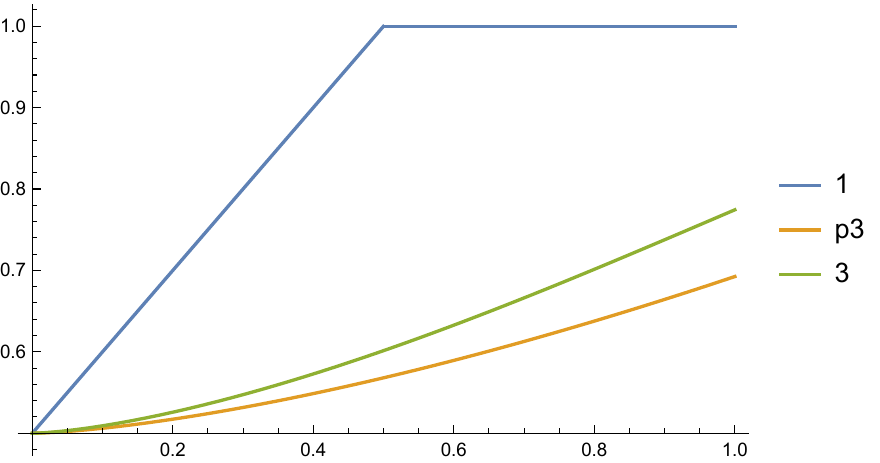}
  \end{center}  
  \caption{A concave (left) and a convex (right) learning profile for a group of three voters compared with a single voter}\label{lin-con}
\end{figure}

A case which instead shifts the long term advantage towards a larger group of voters is the set of profiles for which the individual competence does not converge to 1.  As a simple example we can consider the piecewise linear profiles $p_n(t)=\min(1/2 + a_n t,2/3)$. The difference between this and our earlier linear profiles is that for large $t$ the competence will plateau  and become $2/3$. The value of $a_n$ determines how quickly a group of $n$ voters reach this maximum individual competence level, but the maximum  value itself does not depend on $n$.     In Figure \ref{lin-bound} we plot the individual and group competences for one and three voters, both with $a_n=1$. Note that since the time axis denotes the total time $T$ used the group of three voter reaches competence $2/3$ at a higher total time than the single voter, but for both groups stop the individual competence stops at $2/3$.  However, for the larger group this is still amplified by the majority vote into a group competence which instead converges to $\frac{20}{27}\approx0.74$,  a clear advantage for the larger group. The absolute size of this advantage  will become larger with increasing group size, as long as the total time is large enough. 

\begin{figure}	
  \begin{center}
  	\includegraphics[width=0.40\textwidth]{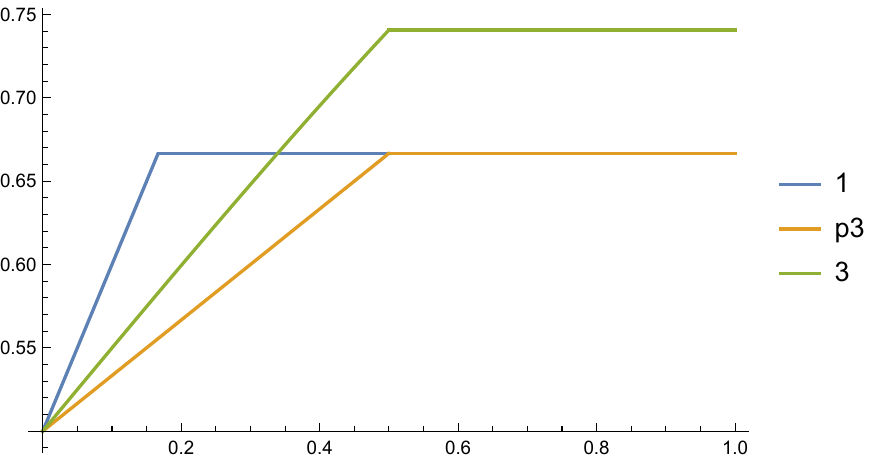}
  \end{center}
  \caption{A learning profile bounded away from 1,  for  groups of  one or three voters, as a function of the total time used.}\label{lin-bound}
\end{figure}

\subsection{Which learning profiles actually occur?}
So far we have demonstrated some of the qualitatively different  behaviours possible under different learning profiles.   This leads to the question of which learning profiles occur in real life cases, and under which circumstances. To a large extent this is an empirical question which should already be present in the research literature, we surmise, though of course not necessarily in the same format as here. We will here present some thoughts regarding which factors might affect the learning profile, and how.

First we note that due to the different types of deliberation involved we would expect to see different types of learning profiles in committee work and larger elections and referenda.   In a well functioning committee the deliberation is typically  quite organised and the members are chosen so that they complement each other's background competencies.  In large scale elections the electorate does not typically depend  on the question to be voted on, and deliberation is not organised in the same sense as for a committee. There are several distinct factors which appear in this description.

First, we can ask more generally how the group size affect the learning profile. In a small group deliberation is relatively easy to organise and one can assure, for example, that all members of the group are both heard and have access to all the other members. As the group size grows communication becomes costlier to organise, and at very large scales will even require physical infrastructure in order to function. So, unless adequate organisation and infrastructure are present we would expect each added member to contribute less and less to the learning rate, when the group size has reached above some threshold.

Secondly,  we can here reconnect to the initial discussion of heterogeneity and homogeneity in deliberation.  In a group which manages to leverage the different background competencies of the members we would expect the group to quickly make early progress and increase the effective competence of the group, hence leading to a learning profile with a rapid increased for small times.  Here we could e.g. see something much like the concave learning profiles in the previous section.

Apart from how these general features affect the learning profiles we can also consider dynamical models for how the individual competence in a group develops over time. Even quite simple such models can lead to both diverse and complicated behaviour.  Let us look at two very simple models.

In our first model we have three voters, indexed by $i=1,2,3$,  which start out at some competence $p_i(0)$ at time $t=0$.  For voter 1 the competence has time derivative which is $0.1(1-p_1(t))$. This means that voter 1 increases their competence, at a rate which slows down as the competence approaches 1.   The other two voters have a time derivative which is equal to the difference between their own   competence and the mean competence of the group. So, these voters drift towards the mean competence.  In Figure \ref{mean-drift} we display the three individual competences and the resulting learning profile of the group, in red.  As we can see the whole group is gradually improving thanks to being lifted by voter 1 which does not simply move towards the average.  Note that the competence of voter 2 initially decreases due to the  low initial mean competence, but eventually all voters have an increasing competence.

 \begin{figure}	
  \begin{center}
  	\includegraphics[width=0.40\textwidth]{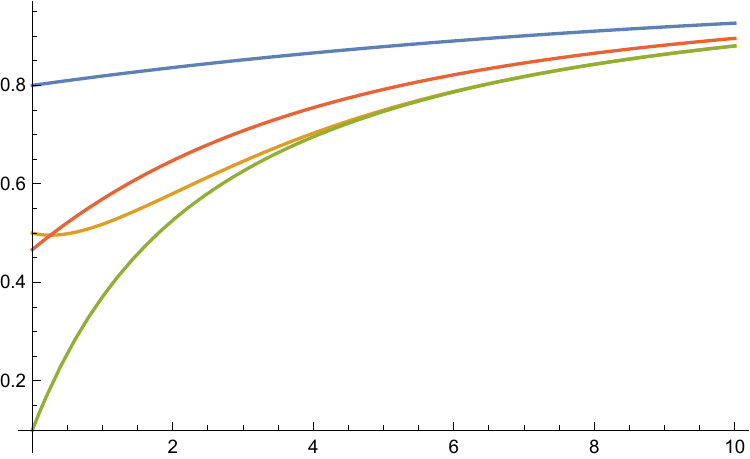}
  \end{center}
  \caption{The individual and group (red) competence in the mean drift model}\label{mean-drift}
\end{figure}

Let us now consider the same model with one modification. Instead of drifting towards the global mean competence each voter, except voter 1,  instead drifts towards the mean competence of those voters which are close enough to their own competence, giving an interaction window around each voter.  In Figure \ref{mean-drift-win} we display the behaviour of this model with four voters and three small modifications.    In the first case we see a behaviour which is similar to our previous model. Some of the voters   have initially decreasing competencies but the voters are spread evenly enough for the positive influence from voter 1 to affect all other voters and the learning profile increases towards 1.   In the second case, middle figure,  the initial competence of voter 3 is slightly lower than in the first case and due to this  the competence of voter 3 is initially decreasing and voter 2 leaves the window of voter 3, who then instead converges with voter 4 at a low competence. Here the learning profile does no converge to 1 as time increases, due to group stuck at the lower competence.    In the third case, rightmost part of the figure,  the initial competences are the same as in the first case, but the derivative of the competence of voter 1 has been multiplied by 2. So the only difference is that voter 1 learns twice as fast.  Just as in the first case the competence of voter 2 is initially decreasing, but now voter 1 improves so fast that they leave the window of voter 2. As a consequence voters 2 to 4 now converge towards their mean competence, which is just slightly above $0.5$.  So here the improved learning rate of voter 1 led to a fragmentation, again leading to a learning profile which does not converge to 1.

 \begin{figure}	
  \begin{center}
  	\includegraphics[width=0.30\textwidth]{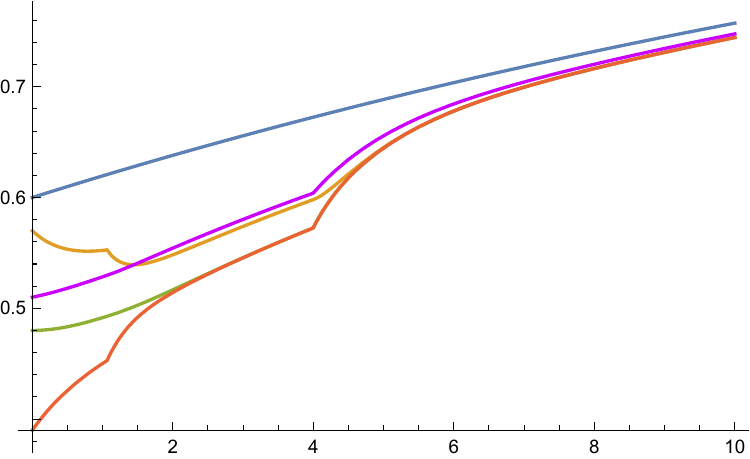}
	\includegraphics[width=0.30\textwidth]{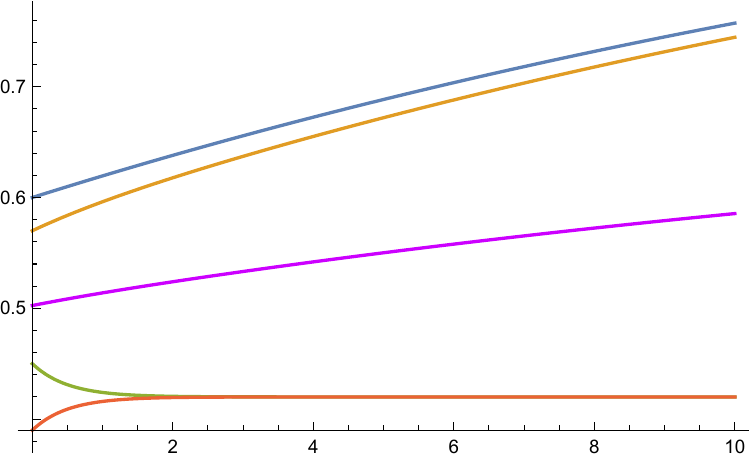}
  	\includegraphics[width=0.30\textwidth]{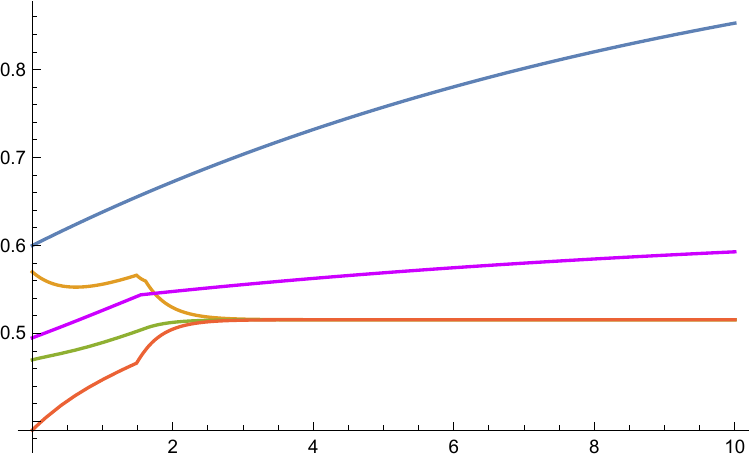}
  \end{center}
  \caption{The individual and group (magenta) competence in the mean drift model with a window}\label{mean-drift-win}
\end{figure}
 
Interestingly, the latter model has some similarities with model which have been studied in opinion dynamics, though the motivation for the mechanism are somewhat different. One such example is the Deffuant model \cite{Deffuant}  in which agents interact pairwise on some network and move towards a common opinion if their current opinions are not too far apart.  Here it has been proven \cite{HT14} that for many networks  there is threshold value for the acceptable opinion difference such that below this the opinions fracture into opinion-clusters, while above this threshold the network converges   to a consensus opinion.  For large scale electorates similar network effects on competence could very well come into play. 

\subsection{Growing competence and correlations}\label{corrsec}
In our discussion so far we have generally assumed that voters are independent, i.e. that there are  no correlations between them apart from what their individual competence levels imply.  However under most realistic models of a deliberative procedure one would expect that interactions during the deliberation will also create some correlations. As we mentioned in our discussion of existing jury theorems, negative correlations are actually beneficial but one might worry that positive correlations will reduce the positive effects of the majority vote.  A detailed discussion and critique of various forms of independence assumptions is given by Dietrich and Spiekerman in \cite{DS17,DS23}.  It is clear that strong enough correlations can significantly lower the probability for a correct vote, but as we saw in Theorem \ref{corrthm} this effect can be controlled in terms of the average strength of the correlations. More generally, the probability for a correct decision is a continuous function of the full probability distribution for the set of votes, in terms of the total variation distance for probability distributions. This means that none of the conclusions we have drawn are  sensitive to the qualitative independence assumption. Rather, the probability for a correct decision will change continuously with the strength of the correlations when such are present.

It is sometimes claimed that strong correlations are one of the main drivers for why majority votes in very large groups are not as near-infallible as the jury theorem might lead one to think.  However, taking examples such as the ones we have just seen into account  one might instead ask if this is not instead, or additionally,  due to having competence levels which are much closer to 0.5 in large groups than in small ones.  We have already mentioned that learning profiles for large groups might increase much slower due to the cost of communication, leading to a lower final competence.  In the past this might to some extent have been compensated for by having correlations which decay rapidly with physical distance within a country, making local communities more or less independent of each other.    Today that positive effect may have been reduced due to rapid communication and  effects of mass-media. The interplay between the learning profile on these large scales and the creation of correlation, by intention or inadvertently,  is certainly worth further scrutiny.

\section{Summary}

In this paper, we have investigated the interaction between the number of voters, the development of their competence over time and the probability for an optimal majority decision. We have developed a model that captures the potentially positive effects of having a heterogeneous group of voters on majority decisions in a more satisfactory way than earlier attempts.

We first considered the situation where we have a fixed total amount of time $T$ and we can either let one voter use the whole amount or instead let $n$ voters use $T/n$ each. This would, for instance, model a situation where we have a given budget for salaries and can spend that budget on either a one member committee, working for a longer period or a larger committee which must finish earlier. Here we showed that with a fixed time budget the larger group performs worse than a single voter unless the larger group actually takes advantage of the group to improve the learning rate, and the improvement in learning rate must be sufficiently large in order outperform a single voter. So, importantly, we do not see any ”wisdom of the crowd” merely by having a bigger crowd; this result shows that communication is essential, even for a crowd. 

Secondly, we considered the cost $C(n,P^{*})$ of achieving a given group competence $P^*$ or different group sizes. Here we showed that the learning profile enters the cost. Due to the linear nature of $P(t)$ for small $t$ it turns out that the cost $C(n,P^*)$ behaves differently depending on whether the learning rate is below, on, or above, the critical learning ratio $c^{*}_n$. For learning rates $c_n$ smaller than $c^*_n$ the cost grows with $n$ and is unbounded; at the critical learning rate the cost converges; and for learning rates larger than the critical one, the cost is decreasing with $n$. Hence, the critical learning rate $c^*_n$ characterises both when a larger group is able to perform better with a given time budget, and when a large group can achieve a given group competence at a lower cost than a single individual.

We also discussed how a single highly efficient expert investigator performs compared to a larger group of non-experts in a setting where the time until a decision must be made is fixed. We also discussed the case where the individual competence is close to 1, in which case the behaviour of the group competence becomes rather simple.

Our results above are for the class of linear learning profiles. As we point out in Section 6, these profiles display  many qualitative features of more general learning profiles, but we do not expect real life profiles to be as simple as that. As we have seen the advantage can shift between smaller and larger groups depending on the shape of the profile. In particular, we always have a long-time advantage for a large group when the individual competence cannot grow above some value which is strictly less than 1.  On the other hand, as the time becomes longer there is also more room for more complicated group dynamics to begin to influence the competence.  As our examples show, even simple models for group dynamics can display very varied behaviours. A longer time span also leaves more room for new correlations to develop and these can potentially offset the advantage of a larger group. 

All together this demonstrates the need for a better empirical understanding of both how the learning profile depends on the setting, and  how  correlations can be kept down in situations where interactions are required for competence growth.

\section*{Acknowledgments}
We are grateful for comments from Franz Dietrich. Financial support from the Swedish Research Council (VR 2015-01588) and Marianne and Marcus Wallenberg’s Foundation (MMW 2015.0084) are gratefully acknowledged.

a\providecommand{\bysame}{\leavevmode\hbox to3em{\hrulefill}\thinspace}
\providecommand{\MR}{\relax\ifhmode\unskip\space\fi MR }
\providecommand{\MRhref}[2]{%
  \href{http://www.ams.org/mathscinet-getitem?mr=#1}{#2}
}
\providecommand{\href}[2]{#2}

\newpage
\section*{Appendix}
\begin{thm}
	Here we let $P_n(\bf{p})$ denote the probability of a majority for the correct outcome and set  $\bar{p}=\frac{1}{n}\sum_{i=1}^n p_i$.
	\begin{enumerate}
		\item  If $p>1/2$ is  fixed  then $P_n({\bf p}(p)) \rightarrow 1$ monotonically with $n$. 
		
		\item   If $ \bar{p}= p> \frac{1}{2} $ for some fixed $p$ then $P_n \rightarrow 1$ and $P_n({\bf p })\geq P_n(\bar{p})$.

		\item  If $\bar{p}=\frac{1}{2}+\frac{\omega(n)}{\sqrt{n}}$, where $\omega(n)$ is any increasing, unbounded, function of $n$, then $P_n \rightarrow 1$.
		
	\end{enumerate}	
\end{thm}
\begin{proof}
	Let $X$ denote the, random, sum of the individual votes
	$$X=\sum_{i=1}^{n}x_i,$$
	where $x_i$ can be 0 or 1. with probabilities $1-p_i$ and $p_i$ respectively. 	Here $x_i=1$ denotes a vote for the correct outcome,  $x_i=0$ the opposite vote, and the group decision is correct if $X>\frac{n}{2}$.
	
	The expectation of $X$  is given by $E[X]=\sum_{i=1}^n p_i=n \bar{p}$  and since the variables are independent Hoeffding's concentration inequality tells us that  $Pr(|X-E[X]|>t)\leq 2exp(-t^2/n)$. So with    $\bar{p}=\frac{1}{2}+\frac{\omega(n)}{\sqrt{n}}$   we find that  $P(X<n/2)<2exp(-\omega^2(n))$, and hence that the probability for a correct majority decision tends to 1  if $\omega(n)\rightarrow\infty$.  This yields both (1) and (3).
	
	Part (2) follows directly from \cite{hoeffding1956}.

\end{proof}

\end{document}